\documentclass[conference,letterpaper]{IEEEtran}

\usepackage[utf8]{inputenc}
\usepackage{amsmath,amssymb,amsfonts}
\usepackage{mathtools}
\usepackage{amsthm}
\usepackage{algorithmic}
\usepackage{graphicx}
\usepackage{textcomp}
\usepackage{tikz}
\usepackage{caption}
\usepackage{cuted}
\usepackage{pgfgantt}
\usepackage{pdflscape}
\usepackage{pst-plot}
\usepackage{comment} 
\usepackage{cases}
\usepackage{nccfoots}
\usepackage{lineno,hyperref}
\usetikzlibrary{spy}
\usetikzlibrary{positioning,calc}
\usetikzlibrary{decorations.pathmorphing,calc,shapes,shapes.geometric,patterns}
\usetikzlibrary{shapes.multipart}
\usepackage{xfrac}
\usepackage{colortbl}
\usepackage{cancel} 
\usetikzlibrary{arrows,positioning,calc,intersections}
\usetikzlibrary{datavisualization.formats.functions}
\def\BibTeX{{\rm B\kern-.05em{\sc i\kern-.025em b}\kern-.08em
    T\kern-.1667em\lower.7ex\hbox{E}\kern-.125emX}}
    \newcommand{\norm}[1]{\left\lVert#1\right\rVert}
\usepackage{romannum}
\usepackage{pgfplots}
\usepgfplotslibrary{fillbetween}
\usetikzlibrary{arrows, decorations.markings}

\newtheorem{claim}{Claim}
\newtheorem{claimproof}{Proof of Claim}
\newtheorem{theorem}{Theorem}
\newtheorem*{theorem*}{Theorem}
\newtheorem{lemma}[theorem]{Lemma}

\newtheorem{definition}[theorem]{Definition}

\newtheorem{remark}[theorem]{Remark}

\newcommand{\setx}{\ensuremath{\mathcal{X}}}
\newcommand{\sety}{\ensuremath{\mathcal{Y}}}
\newcommand{\setz}{\ensuremath{\mathcal{Z}}}

\newcommand{\setu}{\ensuremath{\mathcal{U}}}
\newcommand{\sete}{\ensuremath{\mathcal{E}}}

\newcommand{\mc}{\mathcal}
\newcommand{\mbb}{\mathbb}

\newcommand{\circlearrow}{}
\DeclareRobustCommand{\circlearrow}{%
  \mathrel{\vphantom{\rightarrow}\mathpalette\circle@arrow\relax}%
}
\newcommand{\circle@arrow}[2]{%
  \m@th
  \ooalign{%
    \hidewidth$#1\circ\mkern1mu$\hidewidth\cr
    $#1-$\cr}%
}

\DeclarePairedDelimiterX{\infdivx}[2]{(}{)}{%
  #1\;\delimsize\|\;#2%
}

\begin{document}\onecolumn
\title{Common Randomness Generation from Sources with Countable Alphabet} 



\author{
\IEEEauthorblockN{Wafa Labidi \IEEEauthorrefmark{1}, Rami Ezzine \IEEEauthorrefmark{1}, Christian Deppe\IEEEauthorrefmark{2}\IEEEauthorrefmark{4}, Moritz Wiese\IEEEauthorrefmark{1} and Holger Boche\IEEEauthorrefmark{1}\IEEEauthorrefmark{3}\IEEEauthorrefmark{4}}
\IEEEauthorblockA{\IEEEauthorrefmark{1}Technical University of Munich, Chair of Theoretical Information Technology, Munich, Germany\\
\IEEEauthorrefmark{2}Technical University of Munich, Institute for Communications Engineering,  Munich, Germany\\
\IEEEauthorrefmark{3}CASA -- Cyber Security in the Age of Large-Scale Adversaries–
Exzellenzcluster, Ruhr-Universit\"at Bochum, Germany\\
\IEEEauthorrefmark{4}BMBF Research Hub 6G-life, Munich, Germany\\
Email: \{wafa.labidi, rami.ezzine, christian.deppe, boche\}@tum.de}
}

\maketitle


\begin{abstract}
We study a standard two-source model for common randomness (CR) generation in which Alice and Bob generate a common random variable with high probability of agreement by observing independent and identically distributed (i.i.d.) samples of correlated sources on countably infinite alphabet. The two parties are additionally allowed to communicate as little as possible over a noisy memoryless channel. In our work, we  give a single-letter formula for the CR capacity for the proposed model and provide a rigorous proof of it. This is a challenging scenario because some of the finite alphabet properties, namely of the entropy can not be extended to the countably infinite case. Notably, it is known that the Shannon entropy is in fact discontinuous at all probability distributions with countably infinite support.
\end{abstract}

\section{Introduction}
In the common randomness (CR) generation scheme, the sender and the receiver, often described as terminals, aim to generate a common random variable with high probability of agreement.
The availability of this resource allows to  implement correlated random protocols that often outperform the deterministic ones or the ones using independent randomization in terms of higher speed and efficiency.
One can therefore achieve an enormous performance gain by taking advantage of this resource in the identification scheme, a new approach in communications  developed by Ahlswede and Dueck \cite{Idchannels} in 1989. Indeed, in contrast to the classical transmission scheme proposed by Shannon \cite{Shannon}, the resource CR allows a significant increase in the identification capacity of channels \cite{trafo,part2,ahlswede2021}. In the classical transmission scheme, the transmitter sends a message over the channel and the receiver is interested in decoding the received message. In the identification framework, the encoder sends an
identification message over the channel
and the decoder is interested in knowing whether a specific message of special interest to him has been
sent or not. Naturally, the sender has no knowledge of that message. Otherwise, the problem would be trivial.
The identification approach is much more efficient than the classical transmission scheme for many new applications that impose high requirements on reliability and latency. These applications include machine-to-machine
and human-to-machine systems \cite{application}, digital watermarking \cite{MOULINwatermarking,AhlswedeWatermarking,SteinbergWatermarking}, industry 4.0 \cite{industry4.0} and 6G communication systems \cite{6Gcomm},  Large 6G research projects \cite{researchgroup1}\cite{researchgroup2} are studying the problem of CR generation for future communication networks. 
This is because it is expected that CR will be a highly relevant resource for future communication systems \cite{6Gcomm}\cite{6Gpostshannon}, on the basis of which, the resilience requirements \cite{6Gcomm} and security requirements \cite{semanticsecurity} can be met. The aforementioned requirements are crucial for achieving trustworthiness. It is here worth mentioning that because of modern applications, trustworthiness represents a key challenge for future communication systems \cite{6Gandtrustworthiness}.
Further applications of CR generation include correlated random coding over arbitrarily varying channels (AVCs) \cite{capacityAVC} and oblivious transfer and bit commitment schemes \cite{commitmentcapacity}\cite{unconditionallysecure}
An other obvious application of CR generation is secret key generation, where the generated CR is used as secret keys in cryptography. These keys are used in secure message transmission and message authentication \cite{part1}\cite{maurer}. In this paper, however, no secrecy constraints are imposed.

\color{black}Over the past decades, the problem of CR generation from correlated discrete sources has been investigated in several researches. This problem was initially introduced by 
Ahlswede and Csizár in \cite{part2}. The authors in \cite{part2} considered in particular a two-source model for CR generation, in which the sender and the receiver communicate over a  rate-limited discrete noiseless channels and derived a single-letter formula for the CR capacity for that model.
Later, the authors in  \cite{globecom} extended the results on CR capacity to single-user single-input single-output (SISO) and  Multiple-Input Multiple-Output (MIMO) Gaussian channels, which are highly relevant in many communication situations including satellite and deep space communication links, wired and wireless communications, etc. In addition, the authors in \cite{globecom} proved that the CR capacity of
Gaussian channels is a lower-bound
on the corresponding correlation-assisted secure identification
capacity in the log-log scale in \cite{globecom}. This lower bound can
be greater than the secure identification capacity over Gaussian
channels with randomized encoding derived in \cite{wafapaper}. \color{black} 
Later, the authors in \cite{SISOfasingCR} and in \cite{MIMOfadingCR} focused on 
the problem of CR generation over fading channels, where the concept of outage from the CR generation perspective has been introduced.

Recently, the authors in \cite{isit_CR} studied the problem of CR generation from Gaussian sources and showed that the CR capacity is infinite when the Gaussian sources are perfectly correlated. In such a situation, no communication over the channel is required. 
The major motivation of the work in  \cite{isit_CR} was the drastic effects that the common randomness generated from the perfect feedback in the model treated in \cite{isit_paper} produce on the identification capacity.
The identification capacity of Gaussian channels with noiseless feedback has been established in \cite{isit_paper} and it is infinite regardless of the scaling. The authors in \cite{isit_paper} proposed a coding strategy  that achieves an infinite identification capacity in which an infinitely large amount of CR between the sender and the receiver is generated using noiseless feedback.


However, to the best of our knowledge, very few studies \cite{ResourceEfficientCR} have addressed the problem of CR generation from sources with countably infinite alphabet and as far as we know, no research has focused on deriving the CR capacity for such models.
An example for such source model is the Poisson source model, which is highly useful in molecular communication and optical communication systems. 
The transition to infinite alphabet could have drastic consequences in terms of Shannon entropy convergence, variational distance convergence, etc. Some of the finite alphabet properties, namely of the entropy can not be extended to the countably infinite case. Notably, it has been shown that the Shannon entropy is in fact discontinuous at all probability distributions with countably infinite support \cite{EntropyDisc,VariationalEntropy}.

In our work, we establish a single-letter formula for the CR capacity of a model involving a memoryless source on countably infinite alphabet with one-way communication over noisy memoryless channels. The CR capacity formula established in \cite{part2} for correlated discrete sources is extended to correlated sources on countably infinite alphabet. We use a generalized typicality criterion, called unified typicality \cite{unifiedTyp}, which can be applied to any sources on countable alphabet and make use of the conditional typicality lemma and conditional divergence lemma \cite{unifiedTyp,MarkovLemma} established for the proposed typicality criterion.
\quad \textit{ Paper Outline:} The paper is structured as follows.
In Section \ref{sec:preli}, we recall some auxiliary results related to unified typicality involved in our work. In Section \ref{systemmodelanddefinitions}, 
we present the two-source model for CR generation, provide the key definitions and a single-letter formula for the CR capacity. In Section \ref{sec: directproof}, we provide a rigorous achievability proof of the CR capacity. In Section \ref{converse}, we establish the converse proof of the main result. Section \ref{conclusion} contains concluding remarks and a discussion of some applications of our work.
\section{Preliminaries} \label{sec:preli}
In this section, we briefly present the notation that we adopt in this paper. We also recall some auxiliary results related to unified typicality involved in this work.
\subsection{Notations}
Calligraphic letters $\mathbb{R}, \sety,\setz, \ldots$ are used for finite or infinite sets; lowercase letters $x,y,z,\ldots$ stand for constants and values of random variables; uppercase letters $X,Y,Z,\ldots$ stand for random variables;
For any random variables $X$, $Y$ and $Z$, we use the notation $\color{black}X \circlearrow{Y} \circlearrow{Z}\color{black}$ to indicate a Markov chain.
$\mathbb{R}$ denotes the sets of real numbers; ${D}\infdivx{\cdot}{\cdot}$ denotes the Kullback-Leibler divergence;  
$\norm{\cdot}_2$ denotes the $\ell^2$ norm; $|\cdot|$ denotes the $\ell^1$ norm; 
$P_X$ denotes the probability mass function of a RV $X$ on a finite or countably infinite alphabet; $|\cdot|$ denotes the cardinality of a finite set; the set of probability distributions on the set $\setx$ is denoted by $\mathcal{P}(\setx)$; $H(\cdot)$, $\mathbb{E}(\cdot)$ and $I(\cdot ;\cdot)$ are the entropy, the expected value and the mutual information, respectively;
all logarithms and information quantities are taken to base $2$.
\subsection{Typicality Criteria for Countable Alphabet}
Strong typicality can only be applied to Random Variables (RV)s on finite alphabets \cite{wolfitz}. Thus, results based on strong typicality suffer the same limitation. A unified typicality for finite and countably infinite alphabets has been established in \cite{unifiedTyp}. This typicality concept can be applied to source/channel coding problems on countably infinite alphabet to prove results that cannot be proved by weak typicality. Unified typicality is based on a new information divergence measure introduced in \cite{unifiedTyp}. This typicality unifies both weak typicality \cite{Shannon} and strong typicality \cite{wolfitz}.
\begin{definition}
	Suppose $\nu>0$ and $X^n=(X_1,X_2,\ldots,X_n)$ was emitted by the memoryless source $P_X \in \mathcal{P}(\setx)$ with $\setx$ a countably infinite alphabet and $H(P_X)<\infty$. The unified typical set $\setu_{\nu}^n(P_X)$ w.r.t. $P_X$ is the set of sequences $x^n=(x_1,x_2,\ldots,x_n) \in \setx^n$ such that
	\begin{equation*}
	D\infdivx{Q_X}{P_X}+|H(Q_X)-H(P_X)| \leq \nu,
	\end{equation*},	where $Q_X$ is the empirical distribution of the sequence $x^n$.
	\begin{equation*}
	Q_X(x)=\frac{N(x|x^n)}{n}, \quad \forall x \in \setx,
	\end{equation*}
	where $N(x|x^n)$ is the number of occurrences of $x$ in the sequence $x^n$.
\end{definition}
\begin{remark}
In contrast to the case of finite alphabet, strong typicality does not imply weak typicality when the alphabet is countably infinite. It is known that the Shannon entropy is a continuous function of the probability distribution when the alphabet is finite. However, it is has been proved in \cite{EntropyDisc} that the Shannon entropy is discontinuous at all probability distribution on countably infinite support. By discontinuity of the Shannon entropy, there exist probability distributions defined on countably infinite alphabet that satisfy the strong typicality condition but not the weak typicality condition. For countably infinite alphabets, unified typicality is proved to be stronger than both strong and weak typicality \cite{unifiedTyp}.
\end{remark}
Authors in \cite{unifiedTyp} demonstrated the Asymptotic Equipartition Property (AEP) for unified typicality, which is similar to the AEP for weak and strong typicality. 
\begin{theorem}[\cite{unifiedTyp}]
	Let $H(P_X)$ be finite. For any $\nu>0$: \begin{enumerate}
		\item If $x^n \in \setu_{\nu}^n(P_X)$, then
		\begin{equation*}
		2^{-n(H(P_X)+\nu)} \leq	P_{X^n}(x^n) \leq  2^{-n(H(P_X-\nu))}.
		\end{equation*}
		\item For sufficiently large $n$,
		\begin{equation*}
		\Pr\{X^n \in \setu_\nu^n(P_X)\} > 1-\nu.
		\end{equation*} 
		\item  For sufficiently large $n$,
		\begin{equation*}
		(1-\nu) 2^{n(H(P_X)-\nu)} \leq	|\setu_\nu^n(P_X)| \leq 2^{n(H(P_X)+\nu)},
		\end{equation*}
		where $|\setu_\nu^n(P_X)|$ denotes the cardinality of the set $\setu_\nu^n(P_X)$.
	\end{enumerate}
\label{theorem:Typicality}
\end{theorem}
Unified typicality w.r.t. a bivariate distribution has also been defined in \cite{unifiedTyp}.
\begin{definition}
	Suppose $\nu>0$ and the sequence $(X^n,Y^n)$ was emitted by the memoryless bivariate source $P_{XY} \in \mathcal{P}(\setx\times \sety) $ with $\setx$ and $\sety$ are countably infinite alphabets and {\color{black}{$H(P_{XY})<\infty$}}. The unified jointly typical set $\setu_\nu^{n}(P_{XY})$ w.r.t. $P_{XY}$ is the set of sequences $(x^n,y^n) \in \setx^n \times \sety^n$ such that
	\begin{align*}
	& D\infdivx{Q_{XY}}{P_{XY}} + \vert H(Q_{XY})-H(P_{XY})\vert\\
	& + \vert H( Q_{X})-H(P_{X})\vert + \vert H(Q_{Y})-H(P_{Y})\vert \leq \nu ,
	\end{align*} 
	where $Q_{XY}$ denotes the empirical distribution of the sequence $(x^n,y^n)$.
	\begin{equation*}
	Q_{XY}(x,y)=\frac{N(x,y|x^n,y^n)}{n}, \quad \forall (x,y) \in \setx\times \sety,
	\end{equation*}
	where $N(x,y|x^n,y^n)$ is the number of occurrences of $(x,y)$ in the sequence $(x^n,y^n)$.
\end{definition}
Using the concept of unified typicality, authors in \cite{unifiedTyp} extended the joint asymptotic equipartition property (JAEP) to countably infinite alphabets.
\begin{theorem}[\cite{unifiedTyp}]
Let $H(P_{XY})$ be finite. For any $\nu>0$: \begin{enumerate}
		\item If $(x^n,y^n) \in \setu_{\nu}^n(P_{XY})$, then
		\begin{equation*}
		2^{-n(H(P_{XY})+\nu)} \leq	P^n_{{XY}}(x^n,y^n) \leq  2^{-n(H(P_{XY}-\nu))}.
		\end{equation*}
		\item For sufficiently large $n$,
		\begin{equation*}
		\Pr\{(X^n,Y^n) \in \setu_\nu^n(P_{XY})\} > 1-\nu.
		\end{equation*} 
		\item  For sufficiently large $n$,
		\begin{equation*}
		(1-\nu) 2^{n(H(P_{XY})-\nu)} \leq	|\setu_\nu^n(P_{XY})| \leq 2^{n(H(P_{XY})+\nu)},
		\end{equation*}
		where $|\setu_\nu^n(P_{XY})|$ denotes the cardinality of the set $\setu_\nu^n(P_{XY})$.
	\end{enumerate}
\label{theorem:JointTyp}
\end{theorem}
It has been proved in \cite{unifiedTyp} that unified typicality preserves the consistency property of strong typicality as below.
\begin{theorem}[\cite{unifiedTyp}]
Let $H(P_X)$ and $H(P_Y)$ be finite. For any $\nu>0$, if $(x^n,y^n) \in \setu_{\nu}^n(P_{XY}) $, then $x^n \in \setu_{\nu}^n(P_X)$ and $y^n \in \setu_{\nu}^n(P_Y)$. \label{theorem:consistency}
\end{theorem}
Unified joint typicality can be viewed as a special case of the usual unified typicality, where the sequence $(X,Y)$ is considered as a single RV $Z$. An interesting case is when the sequences $\tilde{X}^n$ and $\tilde{Y}^n$ are output by the statistically independent sources $P_X$ and $P_Y$, respectively. We prove the following Lemma based on Theorem \ref{theorem:Typicality} and Theorem \ref{theorem:JointTyp}.
\begin{lemma}
Let $0<\nu^\prime<\nu$. Suppose that the sequences $\tilde{X}^n$ and $\tilde{Y}^n$ are output by the {\emph statistically independent} sources $P_X$ and $P_Y$, respectively. For any $\nu>\nu^\prime>0$, the probability that $(\tilde{X}^n,\tilde{Y}^n) \in \setu_{\nu}^n(P_{XY})$ for some joint distribution $P_{XY}$ with marginals $P_X$ and $P_Y$ is bounded by
\begin{equation*}
(1-\nu) 2^{- n (I(X;Y)+2\nu^\prime+\nu)}\leq \Pr \bigg \{(\tilde{X}^n,\tilde{Y}^n) \in \setu_{\nu}^n(P_{XY}) \bigg \} \leq 2^{-n(I(X;Y)-2\nu^\prime-\nu)}.
\end{equation*}
\label{lemma:joint}
\end{lemma}
\begin{proof}
\begin{align*}
\Pr \bigg \{(\tilde{X}^n,\tilde{Y}^n) \in \setu_{\nu}^n(P_{XY}) \bigg\} &= \sum_{(\tilde{x}^n,\tilde{y}^n) \in \setu_\nu^n(P_{XY})} P_X^n(\tilde{x}^n) P_Y^n(\tilde{y}^n) \\
& \overset{(a)}{\geq}  |\setu_\nu^n(P_{XY})| 2^{-n (H(P_X)+\nu^\prime)} 2^{-n (H(P_y)+\nu^\prime)} \\
& \overset{(b)}{\geq} (1-\nu) 2^{n(H(P_{XY}-\nu))}  2^{-n (H(P_X)+\nu^\prime)} 2^{-n (H(P_y)+\nu^\prime)} \\
& = (1-\nu) 2^{-n (I(X;Y)+2\nu^\prime+\nu)},
\end{align*}	
where $(a)$ follows from Theorem \ref{theorem:consistency} and Theorem \ref{theorem:Typicality} and $(b)$ follows from Theorem \ref{theorem:JointTyp}.
Similarly, we have
\begin{align*}
\Pr \bigg \{(\tilde{X}^n,\tilde{Y}^n) \in \setu_{\nu}^n(P_{XY}) \bigg\} & \leq 2^{-n(I(X;Y)-2\nu^\prime-\nu)}.
\end{align*} 
\end{proof}
A generalization to a multivariate distribution can be easily shown \cite{unifiedTyp}. In the following, we consider a trivariate distribution.
\begin{definition}
	Suppose $\nu>0$ and the sequence $(X^n,Y^n,Z^n)$ was emitted by the memoryless multivariate source $P_{XYZ} \in \mathcal{P}(\setx\times \sety \times \setz) $ with $\setx$, $\sety$ and $\setz$ countably infinite alphabets and {\color{black}{$H(P_{XYZ})<\infty$}}. The unified jointly typical set $\setu_\nu^{n}(P_{XYZ})$ w.r.t. $P_{XYZ}$ is the set of sequences $(x^n,y^n,z^n) \in \setx^n \times \sety^n \times \setz^n$ such that
	\begin{align*}
	& D\infdivx{Q_{XYZ}}{P_{XYZ}} + \vert H(Q_{XYZ})-H(P_{XYZ})\vert \\
	& +
	\vert H(Q_{XY})-H(P_{XY})\vert+  \vert H(Q_{XZ})-H(P_{XZ})\vert \\
	&+  \vert H(Q_{YZ})-H(P_{YZ})\vert+ \vert H( Q_{X})-H(P_{X})\vert\\
	& + \vert H(Q_{Y})-H(P_{Y})\vert +\vert H(Q_{Z})-H(P_{Z})\vert \leq \nu ,
	\end{align*} 
	where $Q_{XYZ}$ denotes the empirical distribution of the sequence $(x^n,y^n,z^n)$.
\end{definition}
Based on the unified typicality criterion, authors in \cite{MarkovLemma} introduced the following  Markov lemma for countable alphabets.
\begin{theorem}[\cite{MarkovLemma}]
	Let $P_{UXY} \in \mathcal{P}(\setu \times \setx \times \sety)$ be a memoryless multivariate source with $\setu$, $\setx$ and $\sety$ are countable alphabets and $H(P_{UXY}) < \infty$. We assume that $U\circlearrow X\circlearrow Y$ is a Markov chain and 
	\begin{equation}
	\sum_{u}P_{U|X}(u\vert x)(\log P_{U|X}(u\vert x))^{2}< c, \label{conditionVar}
	\end{equation} 
	where the constant $c$ is positive and finite. If for any $\nu>0$ and any given $(x^n,y^n) \in \setu_\nu^n(P_{XY})$, $U^n$ is drawn from $\prod_{i=1}^n P_{U_i|X_i}$, then
	\begin{equation*}
	\Pr \left \{(U^n,x^n,y^n) \in \setu_\nu^n(P_{UXY}) \right\} \geq 1-\nu,
	\end{equation*}
	for sufficiently large $n$ and sufficiently small $\nu$
	\label{theorem:Markov}
\end{theorem}
\section{System Model, Definitions and Main Result}
\label{systemmodelanddefinitions}
In this section, we introduce our system model and extend the definition of an achievable CR rate to the introduced system model presented in Fig. \ref{fig:System}. We then present the main result of this paper.
\subsection{System Model and Definitions}
\label{systemmodel}
Let a memoryless source $P_{XY}$ with two components and variables $X$ and $Y$ on the {\emph{countable alphabets} (finite and countably infinite) $\setx$ and $\sety$, respectively, be given. For instance, a Poisson source is defined on a countably infinite alphabet. The marginal distributions $P_X$ and $P_Y$ satisfy
\begin{equation}
\mathbb{E}\left[\log^2P_X(X)\right], \ \mathbb{E}\left[\log^2P_Y(Y)\right] < \infty. \label{eq: ConstraintOnXY}
\end{equation}
Only Terminal $A$ has access to the source output $X^n$ and only Terminal $B$ has access to the source output $Y^n$. Both outputs have the same block length $n.$ We assume that the joint distribution $P_{XY}$ is known to both terminals. A one-way communication from Terminal $A$
to Terminal $B$ occurs over a memoryless channel $W$. Let $C(W)$ denote the Shannon capacity of the channel $W$. No other resources are available to any of the terminals. This is the standard two-source model introduced by Ahlswede and Csiszár in \cite{part2}, where they considered the communication over a discrete memoryless noiseless channel with limited capacity.
A CR-generation protocol of block length $n$ as introduced in \cite{part2} is composed of:
\begin{enumerate}
    \item a function $\Phi$ that assigns to the random sequence $X^n$ a RV $K$ generated by the terminal $A$ and defined on the alphabet $\mathcal{K}$ with $|\mathcal{K}|\geq 3$,
    \item a function $\Lambda$ that converts $X^n$ into the input sequence $T^n$,
    \item a function $\Psi$ that assigns to $Y^n$ and the output sequence $Z^n$ a RV $L$ generated by Terminal $B$ and defined on the alphabet $\mathcal{K}$ .
\end{enumerate}
This protocol outputs a pair of RVs $(K,L)$. This pair $(K,L)$ is called permissible \cite{part2} if the following is satisfied.
\begin{equation}
    K=\Phi(X^{n}), \ \     L=\Psi(Y^{n},Z^{n}).
    \label{KLSISOcorrelated}
\end{equation}
The system model is depicted in Fig. \ref{fig:System}.
\begin{figure}[htb!]
\centering
\tikzstyle{block} = [draw, rectangle, rounded corners,
minimum height=2em, minimum width=2cm]
\tikzstyle{blockchannel} = [draw, top color=white, bottom color=white!80!gray, rectangle, rounded corners,
minimum height=1cm, minimum width=.3cm]
\tikzstyle{input} = [coordinate]
\usetikzlibrary{arrows}
\scalebox{.9}{
\begin{tikzpicture}[scale= 1,font=\footnotesize]
\node[blockchannel] (source) {$P_{XY}$};
\node[blockchannel, below=3cm of source](channel) {noisy memoryless channel $W$};
\node[block, below left=3.2cm of source] (x) {Terminal $A$};
\node[block, below right=3.2cm of source] (y) {Terminal $B$};
\node[above=1cm of x] (k) {$K=\Phi(X^n)$};
\node[above=1cm of y] (l) {$L=\Psi(Y^n,Z^n)$};

\draw[->,thick] (source) -- node[above] {$X^n$} (x);
\draw[->, thick] (source) -- node[above] {$Y^n$} (y);
\draw [->, thick] (x) |- node[below right] {$T^n=\Lambda(X^n)$} (channel);
\draw[<-, thick] (y) |- node[below left] {$Z^n$} (channel);
\draw[->] (x) -- (k);
\draw[->] (y) -- (l);

\end{tikzpicture}}
\caption{Bivariate memoryless countable-alphabet source model with one-way communication over a noisy memoryless channel}
\label{fig:System}
\end{figure}
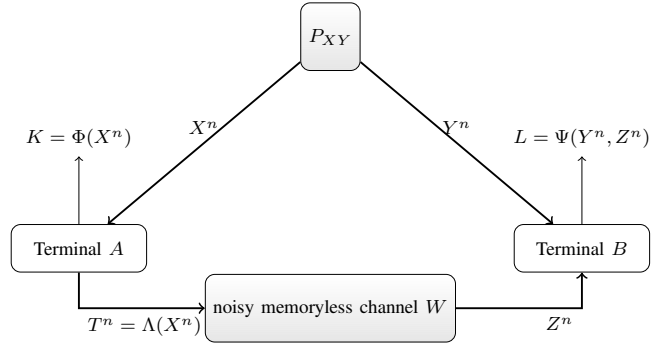
Now, we provide the definition of an achievable CR rate and of the CR capacity w.r.t. our system model depicted in Fig. \ref{fig:System}. We extend the definition of an achievable CR rate and CR capacity introduced in \cite{part2} to our system model.
\begin{definition}  We call a number $H$ an achievable CR rate for the system model in Fig. \ref{fig:System} if there exists a non-negative constant $c$ such that for every $\epsilon>0$ and $\gamma>0$ and for sufficiently large $n$ there exists a permissible  pair of RVs $(K,L)$ such that
\begin{align}
    \Pr\{K\neq L\}\leq \epsilon, 
    \label{errorcorrelated}
\end{align}
\begin{align}
    |\mathcal{K}|\leq 2^{cn},
    \label{cardinalitycorrelated}
\end{align}
\begin{align}
    \frac{1}{n}H(K)> H-\gamma.
     \label{ratecorrelated}
\end{align}
\end{definition}
Now, we extend the definition of the CR capacity introduced in \cite{part2} to our system model depicted in Fig. \ref{fig:System}.
\begin{definition} 
The CR capacity $C_{CR}(p_{XY},W)$ for the system model in Fig. \ref{fig:System} is the maximum achievable CR rate.
\end{definition}
Now, we present the main result of our work. We characterize the CR capacity of the system model in Fig. \ref{fig:System}.
\begin{theorem}
For the system model depicted in Fig. \ref{fig:System}, the CR capacity $C_{CR}(P_{XY},W)$ is given by
\begin{align}
	C_{CR}(P_{XY},W)&=
	\underset{ \substack{U \in \mathcal{U} \\{\substack{U \circlearrow{X} \circlearrow{Y}\\ I(U;X)-I(U;Y) \leq \color{black}C(W)\color{black}}}}}{\sup} I(U;X), \nonumber
	\end{align}
	where ${\color{black}{C(W)}}$ is the Shannon capacity of $W$ and the set $\mathcal{U}$ is defined as follows
	 \begin{equation}
	 \mathcal{U}=\Big\{U: \quad \mathbb{E}\left[ \log^2(P_{U|X}(U|X=x))|X=x \right]<\infty,\ \forall x \in \setx \Big\}. \label{eq:SetU}
	 \end{equation}
	 \label{mainTheorem}
\end{theorem} 
\begin{remark}
In contrast to Shannon message transmission, CR shows a performance gain in terms of rate within the identification scheme. We can increase the identification capacity by taking advantage of correlation \cite{corrDMC}. For this reason, CR generation for future communication networks is a central research question in large 6G research projects \cite{researchgroup1}\cite{researchgroup2}. The goal within these projects is to experimentally demonstrate the applications of CR generation in 6G communication systems. In particular, use cases for CR generation are being considered when no communication over the channel is necessary. It is also worth mentioning that CR is highly relevant in the modular coding scheme for secure communication, where CR can be used as a seed \cite{semanticsecurity}. CR is a useful resource for coding over AVCs because we require only a little amount of CR compared to the set of messages.
\end{remark}
\section{Direct Proof of Theorem \ref{mainTheorem}} \label{sec: directproof}
In this section, we prove the direct part of Theorem \ref{mainTheorem}. The proof proceeds, to some extent, in the same way as in in \cite{part2}, where the code construction is based on the same type of binning as for the Wyner-Ziv problem. Instead of strong typicality, we use the concept of unified typicality in the encoding/decoding metrics and in the error probability analysis.

We first justify the use of the concept of unified typicality in the encoding/decoding metrics. It is easy to verify that \eqref{eq: ConstraintOnXY} and the definition of the set $\mathcal{U}$ in \eqref{eq:SetU} imply that
\begin{equation}
H(P_X), H(P_Y), H(P_{U|X}) < \infty. \label{eq: InfiniteInputEntropy}
\end{equation}
It follows from \eqref{eq: InfiniteInputEntropy} that
\begin{align}
H(P_U) & = H(P_{UX})-H(P_{X|U}) \nonumber\\
& \leq  H(P_{UX}) \nonumber \\
& = H(P_{X}) + H(P_{U|X}) \nonumber\\
& < \infty.
\end{align}
Since $H(P_{U})$, $H(P_{X})$ and $H(P_Y)$ are finite, all possible combinations of joint entropy are finite. Therefore, we can apply Theorem \ref{theorem:Typicality} and Theorem \ref{theorem:JointTyp} on marginal and joint probability distributions, respectively. 
Let $U$ be an arbitrary random variable on $\setu$ satisfying $U \circlearrow{X} \circlearrow{Y}$ and  $I(U;X)-I(U;Y) \leq C(W)$.
In the following, we show that $H=I(U;X)$ is an achievable CR rate for our system model. Let the probability distribution $P_{U|X}$ be given. Let $0<\nu<\nu_1<\nu_2<\nu_3$. \\

{\bf{Code Construction}}: We generate $N_1N_2$ codewords $U^n(i,j),\quad i=1,\ldots,N_1,\ j=1,\ldots,N_2$ by choosing the $n.(N_1N_2)$ symbols $u_l(i,j)$, $l=1,\ldots,n$, independently at random using $P_U$ (computed from $P_{XU}$). Without loss of generality, assume that the
distribution of $U$ is a possible type for block length $n$. Each realization $u^n(i,j)$ of $U^n(i,j)$ is known to Terminal $A$ and Terminal $B$. For some $\delta> \frac{3}{2} \nu_1$, let
{{\begin{align}
		N_{1}&=2^{\left(n[I(U;X)-I(U;Y)+4\delta]\right)} \nonumber, \\
		N_{2}&=2^{\left(n[I(U;Y)-2\delta]\right)}. \nonumber
		\end{align}}}
	
{\bf{Encoder}}: Let $(x^n,y^n)$ be any realization of $(X^n,Y^n)$. Given $x^n$ with $(x^n,y^n) \in \setu_{\nu_1}^n(P_{XY})$, the task of the encoder consists in finding a pair $(i,j)$ such that $\left(x^n,u^n(i,j) \right)$ are jointly typical. If such a pait $(i,j)$ exists, we set $f(x^n)=i$. If not successful, then $f({x}^n)$ is set to $N_1+1$ and $K$ is set to a constant sequence ${u}^n_0$ different from all the ${u}^n(i,j)$s and known to both terminals. We choose $\nu$ to be sufficiently small such that
\begin{align}
\frac{\log \lVert f \rVert}{n}&=\frac{\log(N_1+1)}{n} \nonumber \\
&\leq C(W)-\delta^\prime,
\label{inequalitylogfSISO}
\end{align}
for some $\delta^\prime$, where $\lVert f \rVert$ denotes the cardinality of the codomain of the function $f$. The message $i^{\star}=f({x}^n)$, with $i^{\star}\in\{1,\ldots,N_1+1\}$, is mapped into a sequence ${t}^n$ using a suitable \textit{forward error correcting code} with rate $\frac{\log \lVert f \rVert}{n}$ satisfying \eqref{inequalitylogfSISO} and with error probability lower than $\frac{\epsilon}{2}$ for sufficiently large $n$. The sequence ${t}^n$ is transmitted over the channel $W$. \\

{\bf Decoder}: Let ${z}^n$ denote the channel output sequence. Terminal $B$ converts the channel output $z^n$ into $\hat{i}^{\star}$. From the knowledge of $\hat{i}^{\star}$ and $y^n$, the task of the decoder is to find ${j}$ such that $\left(u^n(\hat{i}^{\star},{j}),y^n\right)$ are jointly typical. If such an index $j$ exists, let $L({y}^n,\hat{i}^{\star})={u}^n(\hat{i}^{\star},{j})$. If there is no such ${u}^n(\hat{i}^{\star},{j})$ or there are several, $L$ is set to ${u}^n_0$.\\

{\bf Error Analysis}:  We consider the following error events. 
\begin{enumerate}
	\item Suppose that $(x^n,y^n)$ are not jointly typical:
	\begin{equation*}
	\sete_1=\left\{ (X^n, Y^n) \notin \setu_{\nu_1}^n(P_{XY}) \right\}.
	\end{equation*}
 \item Suppose that $(x^n,y^n) \in \setu_{\nu_1}^n(P_{XY})$ but the encoder fails to find a pair $(i,j)$ such that $\left(u^n(i,j),x^n\right) \in \setu_{\nu_2}^n(P_{UX}) $:
 \begin{equation*}
 \sete_2= \bigcap_{\substack{i=1,\ldots,N_1 \\j=1,\ldots,N_2}} \big\{ \left(U^n(i,j),X^n\right) \notin \setu_{\nu_2}^n(P_{UX}) \big\}.
 \end{equation*}
 \item Suppose that $(x^n,y^n) \in \setu_{\nu_1}^n(P_{XY})$ and the encoder outputs a pair $(i,j)$ such that $\left(u^n(i,j),x^n\right) \in \setu_{\nu_2}^n(P_{UX})$. However, the decoder outputs $\tilde{j}\neq j$ such that $\left(u^n(\hat{i},\tilde{j}),y^n\right) \in \setu_{\nu_2}^n(P_{UY}) $:
 \begin{equation*}
 \sete_3=\cup_{\substack{\tilde{j}=1,\ldots,N_2 \\ \tilde{j}\neq j}} \left\{ \left(U^n(\hat{i},\tilde{j}),Y^n\right) \in \setu_{\nu_2}^n(P_{UY}) \right\}.
 \end{equation*}
 \item  Suppose that $(x^n,y^n) \in \setu_{\nu_1}^n(P_{XY})$ and the encoder oututs a pair $(i,j)$ such that $\left(u^n(i,j),x^n\right) \in \setu_{\nu_2}^n(P_{UX})$. However, the decoder fails to find ${j}$ such that $\left(u^n(\hat{i},{j}),x^n,y^n\right) \in \setu_{\nu_3}^n(P_{UXY})$:
 \begin{equation*}
 \sete_4 = \bigcap_{j=1,\ldots,N_2}\left\{ \left(U^n(\hat{i},{j}),x^n,y^n\right) \notin \setu_{\nu_3}^n(P_{UXY}) \right\}.
 \end{equation*}
\end{enumerate}
Let $P_e$ denote the probability of the overall error event. 
\begin{equation*}
P_e \leq \Pr\{\sete_1\}+\Pr\{\sete_2\}+\Pr\{\sete_3\}+\Pr\{\sete_4\}.
\end{equation*}
An upper-bound for the overall error probability $P_e$ is established in the following.
\begin{align*}
\Pr\{\sete_1\}&=P_{XY}^n\left((\setu_{\nu_1}^n(P_{XY}))^c \right) \\
& =1-P_{XY}^n\left( \setu_{\nu_1}^n(P_{XY})\right)\\
&\overset{(a)}{\leq}\nu_1,
\end{align*}
where $(a)$ follows from Theorem \ref{theorem:JointTyp} since the sequence $(x^n,y^n)$ is drawn from the distribution $P_{XY}^n$. Note that $H(P_{XY})$ is finite.
\begin{align*}
\Pr\{\sete_2\} & =  P_X^n({\setu_\nu^n(P_{X})^c}) \Pr\big\{\sete_2|X^n \notin \setu_\nu^n(P_{X}) \big \} +  P_X^n({\setu_\nu^n(P_{X})})\Pr\big \{\sete_2|X^n \in \setu_\nu^n(P_{X}) \big \} \\
& {\leq} P_{X}^n\left((\setu_\nu^n(P_{X}))^c\right)
+  \Pr\left\{\sete_2|X^n \in \setu_\nu^n(P_{X}) \right \} \\
&\overset{(a)}{\leq} \nu + \prod_{\substack{i=1,\ldots,N_1 \\j=1,\ldots,N_2}} \Pr\bigg\{ \big(U^n(i,j),X^n \big) \notin \setu_{\nu_2}^n(P_{UX})|X^n \in \setu_\nu^n(P_{X}) \bigg\} \\
&\overset{(b)}{\leq} \nu+  \bigg(1-(1-\nu_2)2^{-n\big(I(U,X)+3\nu_2 \big)}\bigg)^{N_1N_2} \\
& \overset{(c)}{\leq} \nu+  \exp\left(-N_1N_2(1-\nu_2)2^{-n\big(I(U,X)+3\nu_2\big)}\right)\\
&\overset{(d)}{\leq} \nu+  \exp\left(-(1-\nu_2)2^{n(2\delta-3\nu_2)} \right)\\
&\overset{(e)}{\leq} \nu,
\end{align*}
where $(a)$ follows because the $N_1N_2$ events of the intersection are independent and from Theorem \ref{theorem:Typicality}, $(b)$ follows from Theorem \ref{theorem:JointTyp}, $(c)$ follows because $(1-x)^m\leq \exp(-mx)$, $(d)$ follows from the definition of $N_1$ and $N_2$ and $(e)$ follows because $\exp\left(-(1-\nu_2)2^{n(2\delta-3\nu_2)} \right)$ goes to zero when $n$ goes to infinity.

\begin{align*}
\Pr\{\sete_3\}& \overset{(a)}{\leq} \sum_{\tilde{j}\neq j} \Pr\left\{\left(U^n(\hat{i},\tilde{j}),Y^n\right) \in \setu_{\nu_2}^n(P_{UY})\right \} \\
& \overset{(b)}{<} N_2 \cdot 2^{-n(I(U,Y)-3{\nu_2})}, \\
& = 2^{-n (2\delta-3{\nu_2})}\\
&=0,\quad n\to \infty,
\end{align*}
where $(a)$ follows from the union bound and $(b)$ follows from Theorem \ref{lemma:joint}. 
\begin{align*}
\Pr\{\sete_4 \} &=\Pr \left[ \cap_{j=1,\ldots,N_2}\left\{ \left(U^n(\hat{i},{j}),x^n,y^n\right) \notin \setu_{\nu_3}^n(P_{UXY}) \right\}  \right] \\
& \overset{(a)}{=} \prod_{j=1}^{N_2} \Pr \left\{ \left(U^n(\hat{i},{j}),x^n,y^n\right) \notin \setu_{\nu_3}^n(P_{UXY})  \right\} \\
& \overset{(b)}{\leq} \nu_3^{N_2}\\
&= 0, \quad n \to \infty,
\end{align*}
where $(a)$ follows because the $N_2$ events of the intersection are independent and $(b)$ follows from Theorem \ref{theorem:Markov}. Note that $H(P_{UXY})$ is finite and the assumption \eqref{conditionVar} is satisfied.
Finally, when $n$ goes to infinity, the average error probability $P_e$ goes to zero.
\begin{align}
P_e  & =\sum_{i=1}^{4} \Pr\{\sete_i\} \nonumber\\
& \leq \nu+\nu_1 \\
&< \frac{\epsilon}{2}. \label{eq:overallError}
\end{align}
Let $I^\star=f(X^n)$ be the RV modeling the message encoded by Terminal $A$ and let $\hat{I}^\star$ be the RV modeling the message decoded by Terminal $B$. 
	We have:
	\begin{align}
	\Pr\{K\neq L\} 
	&\leq \Pr\{K\neq L|I^\star=\hat{I}^\star\}+ \Pr\{I^\star\neq\hat{I}^\star\}\nonumber\\
	&\overset{(a)}{\leq} P_e+\frac{\epsilon}{2} \nonumber\\
	& \overset{(b)}{\leq}  \epsilon, \nonumber 
	\end{align}
	where $(a)$ follows from the union bound and $(b)$ follows from \eqref{eq:overallError}. Thus, the pair $(K,L)$ satisfies \eqref{errorcorrelated}.
It remains to show that $(K,L)$ satisfy \eqref{cardinalitycorrelated} and \eqref{ratecorrelated}.
Clearly, \eqref{cardinalitycorrelated} is satisfied  for $c=2\left[I(U;X)+2\delta\right]$, $n$ sufficiently large:
\begin{align}
|\mathcal{K}|&=N_1 N_2+1 \nonumber \\
&= 2^{(n\left[I(U;X)+2\delta\right])}+1 \nonumber \\
&\leq 2^{(2n\left[I(U;X)+2\delta\right])}.\nonumber
\end{align}
For a fixed $u^n(i,j) \in \setu^n$, it holds that
\begin{align*}
&\Pr\{K={u}^n(i,j)\} \\
&\overset{(a)}{=}\sum_{{x}^n\in \setu_{\nu}^n(P_X)}\Pr\{K={u}^n(i,j)|X^n={x}^n\}P_{X}^n({x}^n) \\
&
\overset{(b)}{\leq} 2^{\left(-n(I(U;X)+3\nu_2)\right)}, 
\end{align*}
where $(a)$ follows because for $(x^n,{u}^n(i,j))$ being not jointly typical, we have $\Pr \{K={u}^n(i,j)|X^n={x}^n\}=0$ and $(b)$ follows from Theorem \ref{theorem:JointTyp}. This yields
{{\begin{align}
		H(K) & \geq 
		n (I(U;X)+ 3\nu_2) \nonumber\\
		& = n H+o(n). \nonumber
		\end{align}}}
	Thus, (\ref{ratecorrelated}) is satisfied. 
	This completes the direct proof.
\section{Converse proof}
\label{converse}
Let $(K,L)$ be a permissible pair w.r.t. a fixed CR-generation protocol of block-length $n,$ as presented in Section \ref{systemmodel}. That means that the pair $(K,L)$ satisfies \eqref{errorcorrelated},
 \eqref{cardinalitycorrelated} and \eqref{ratecorrelated}.

We show that any achievable CR rate $H$ satisfies
\begin{align}
H < \underset{ \substack{U \in \mathcal{U} \\{\substack{U \circlearrow{X} \circlearrow{Y}\\ I(U;X)-I(U;Y) \leq \color{black}C(W)\color{black}}}}}{\sup} I(U;X)+\epsilon'', \nonumber 
\end{align}
where 
\begin{equation}
\mathcal{U}=\Big\{U: \quad \mathbb{E}\left[ \log^2(P_{U|X}(U|X=x))|X=x \right]<\infty,\ \forall x \in \setx \Big\} \nonumber
\end{equation}
and where $\epsilon''>0$ is an arbitrarily small positive constant. 

In our proof, we will use Lemma 17.12 in \cite{IT_CiKo}.
Let $J$ be a random variable uniformly distributed on $\{1,\dots, n\}$ and independent of $K$, $X^n$ and $Y^n$. We further define $U=(K,X_{1},\dots, X_{J-1},Y_{J+1},\dots, Y_{n},J).$ It holds that $U \circlearrow{X_J} \circlearrow{Y_J}.$ In what follows, we will show that $U \in \mathcal{U}.$

\begin{claim}
	\label{finitemean}
	For a fixed block-length $n$ and $\forall x \in \mathcal{X}:$
 \begin{align*}
  \color{black}\mathbb{E}\left[\log^{2} P_{U|X_{J}=x}(U|X_{J}=x)|X_{J}=x \right]<\infty. \color{black}
 \end{align*}
\end{claim}
\begin{claimproof}
	We have
 \begin{align}
 	&P_{U|X_{J}=x}(U|X_{J}=x) \nonumber \\
 	&=P_{K,X_{1},\hdots,X_{J-1},Y_{J+1},\hdots,Y_{n},J|X_{J}=x}(K,X_{1},\hdots,X_{J-1},Y_{J+1},\hdots,Y_{n},J|X_{J}=x) \nonumber \\
 	&\overset{(a)}{=}P_{K|X_{1},\hdots,X_{J-1},Y_{J+1},\hdots,Y_{n},J,X_{J}=x}(K|X_{1},\hdots,X_{J-1},Y_{J+1},\hdots,Y_{n},J,X_{J}=x) \left[\prod_{i=1}^{J-1} P_{X}(X_{i})\right] \left[\prod_{i=J+1}^{n} P_{Y}(Y_{i})\right] P_{J}(J) \nonumber \\
 	&\overset{(b)}{=}\frac{1}{n}P_{K|X_{1},\hdots,X_{J-1},Y_{J+1},\hdots,Y_{n},J,X_{J}=x}(K|X_{1},\hdots,X_{J-1},Y_{J+1},\hdots,Y_{n},J,X_{J}=x)\prod_{i=1}^{J-1} P_{X}(X_{i}) \prod_{i=J+1}^{n} P_{Y}(Y_{i}), \nonumber 
 \end{align}
 where $(a)$ follows because $X_i,i=1\hdots n$ are mutually independent and because $J$ is independent of $X^n,$ and $(b)$ follows because $J$ is uniformly distributed on $\{1,\hdots,n\}.$

 Therefore, we have
 \begin{align*}
 	&\log P_{U|X_{J}=x}(U|X_{J}=x)  \nonumber \\
 	&=\log P_{K|X_{1},\hdots,X_{J-1},Y_{J+1},\hdots,Y_{n},J,X_{J}=x}(K|X_{1},\hdots,X_{J-1},Y_{J+1},\hdots,Y_{n},J,X_{J}=x) \nonumber \\
 	& \quad +\sum_{i=1}^{J-1}\log(P_{X}(X_{i}))+\sum_{i=J+1}^{n}\log(P_{Y}(Y_{i}))+\log(\frac{1}{n}).
 \end{align*}
 It follows that
 \begin{align*}
 	&\left(\log P_{U|X_{J}=x}(U|X_{J}=x)\right)^2  \nonumber \\
 	&=\bigg\lvert \log P_{U|X_{J}=x}(U|X_{J}=x)       \bigg\vert^{2} \nonumber \\
 	&\overset{(a)}{\leq} 2\left( \bigg\lvert \log P_{K|X_{1},\hdots,X_{J-1},Y_{J+1},\hdots,Y_{n},J,X_{J}=x}(K|X_{1},\hdots,X_{J-1},Y_{J+1},\hdots,Y_{n},J,X_{J}=x)+\log(\frac{1}{n})\bigg\rvert^{2}\right) \nonumber \\
 	&\quad+2\left(\bigg\lvert\sum_{i=1}^{J-1}\log(P_{X}(X_{i}))+\sum_{i=J+1}^{n}\log(P_{Y}(Y_{i}))\bigg\rvert^{2}\right) \nonumber \\
 	&\overset{(b)}{\leq} 4 \left(  \log^2 P_{K|X_{1},\hdots,X_{J-1},Y_{J+1},\hdots,Y_{n},J,X_{J}=x}(K|X_{1},\hdots,X_{J-1},Y_{J+1},\hdots,Y_{n},J,X_{J}=x)+\log^2(\frac{1}{n})\right) \nonumber \\
 	&\quad+4\left(\bigg\lvert\sum_{i=1}^{J-1}\log(P_{X}(X_{i}))\bigg\rvert^{2}+\bigg\lvert\sum_{i=J+1}^{n}\log(P_{Y}(Y_{i}))\bigg\rvert^{2}   \right) \nonumber \\
 	&\overset{(c)}{\leq} 4 \left(  \log^2 P_{K|X_{1},\hdots,X_{J-1},Y_{J+1},\hdots,Y_{n},J,X_{J}=x}(K|X_{1},\hdots,X_{J-1},Y_{J+1},\hdots,Y_{n},J,X_{J}=x)+\log^2(\frac{1}{n})\right) \nonumber \\
 	&\quad+4\left(\left(\sum_{i=1}^{J-1}\lvert\log(P_{X}(X_{i}))\rvert\right)^{2}+\left(\sum_{i=J+1}^{n}\lvert\log(P_{Y}(Y_{i}))\rvert\right)^{2}\right) \nonumber \\
 	&\overset{(d)}{\leq}  4 \left(  \log^2 P_{K|X_{1},\hdots,X_{J-1},Y_{J+1},\hdots,Y_{n},J,X_{J}=x}(K|X_{1},\hdots,X_{J-1},Y_{J+1},\hdots,Y_{n},J,X_{J}=x)+\log^2(\frac{1}{n})\right) \nonumber \\
 	&\quad+4\left((J-1)\sum_{i=1}^{J-1}\log^{2}(P_{X}(X_{i}))+(n-J)\sum_{i=J+1}^{n}\log^{2}(P_{Y}(Y_{i}))\right), \nonumber 
 \end{align*}
 where $(a)(b)$ follow because $\lvert x+y\rvert^2\leq 2\left(\lvert x \rvert^2+\lvert y\rvert^2\right),$ $(c)$ follows from the triangle's inequality and $(d)$ follows because $\left(\sum_{i=1}^{n}x_i\right)^2\leq n\sum_{i=1}^{n}x_i^2.$
\end{claimproof}
Therefore, it follows that
\begin{align}
	&\mathbb{E} \left[\left(\log P_{U|X_{J}=x}(U|X_{J}=x) \right)^2|X_{J}=x\right] \nonumber \\
	&\leq  4 \left(  \mathbb{E}\left[\log^2 P_{K|X_{1},\hdots,X_{J-1},Y_{J+1},\hdots,Y_{n},J,X_{J}=x}(K|X_{1},\hdots,X_{J-1},Y_{J+1},\hdots,Y_{n},J,X_{J}=x)|X_{J}=x\right]+\log^2(\frac{1}{n})\right) \nonumber \\
	&\quad+4\left((J-1)^{2}\mathbb{E}\left[\log^{2}(P_{X}(X))\right]+(n-J)^2\mathbb{E}\left[\log^{2}(P_{Y}(Y))\right]\right) \nonumber \\
	&\leq  4 \left(  \mathbb{E}\left[\log^2 P_{K|X_{1},\hdots,X_{J-1},Y_{J+1},\hdots,Y_{n},J,X_{J}=x}(K|X_{1},\hdots,X_{J-1},Y_{J+1},\hdots,Y_{n},J,X_{J}=x)|X_{J}=x\right]+\log^2(\frac{1}{n})\right) \nonumber \\
	&\quad+4\left(n^2\left(\mathbb{E}\left[\log^{2}(P_{X}(X))\right]+\mathbb{E}\left[\log^{2}(P_{Y}(Y))\right]\right)\right) \label{bound1}
\end{align}
Since by assumption $\mathbb{E}\left[\log^{2}(P_{X}(X))\right]$ and $\mathbb{E}\left[\log^{2}(P_{Y}(Y))\right]$ are both finite, it remains to prove that $\mathbb{E}\left[\log^2 P_{K|X_{1},\hdots,X_{J-1},Y_{J+1},\hdots,Y_{n},J,X_{J}=x}(K|X_{1},\hdots,X_{J-1},Y_{J+1},\hdots,Y_{n},J,X_{J}=x)|X_{J}=x\right]$ is finite.
It holds using the law of total expectation that
\begin{align}
&\mathbb{E}\left[\log^2 P_{K|X_{1},\hdots,X_{J-1},Y_{J+1},\hdots,Y_{n},J,X_{J}=x}(K|X_{1},\hdots,X_{J-1},Y_{J+1},\hdots,Y_{n},J,X_{J}=x)|X_{J}=x\right] \nonumber \\
&=\sum_{x_1, ,\hdots,x_{j-1},y_{j+1},\hdots,y_{n},j} P_{X_1, ,\hdots,X_{j-1},Y_{j+1},\hdots,Y_{n},J|X_J=x}(x_1, ,\hdots,x_{j-1},y_{j+1},\hdots,y_{n},j|X_J=x) \nonumber \\
&\quad  \times \mathbb{E} \bigg [\log^2 P_{K|X_{1},\hdots,X_{J-1},Y_{J+1},\hdots,Y_{n},J,X_{J}=x}(K|X_{1},\hdots,X_{J-1},Y_{J+1},\hdots,Y_{n},J,X_{J}=x)|X_1=x_1, ,\hdots,X_{J-1}=x_{j-1} \nonumber
\\ & \quad ,Y_{J+1}=y_{j+1},\hdots,Y_{n}=y_{n},J=j,X_{J}=x\bigg]
\end{align}

Consider $S=(X_1,\hdots, X_{J-1},Y_{J+1},\hdots,Y_{n},J,X_J)$ 
and let $s=(x_1,\hdots, x_{j-1},y_{j+1},\hdots,y_{n},j,x)$ be any realization of $S$ 

\begin{lemma}
	\label{upperboundvariance}
	For $\lvert \mc K \rvert \geq 3,$ it holds for sufficiently large $n$ that
	\begin{align}
	\mbb E\left[\log^{2}P_{K|S=s}(K|S=s)|S=s\right]<\infty.
	\nonumber \end{align}
\end{lemma}
\begin{proof}
	We have
	\begin{align}
	\mbb E \left[\log^{2} P_{K|S=s}(K|S=s)|S=s\right]
	=\frac{1}{\ln(2)^2}\mbb E \left[\ln^{2} P_{K|S=s}(K|S=s)|S=s\right].
	\nonumber \end{align}
	
	Define the following two sets
	\begin{align}
	\mc K_{L}(s)=\{k\in \mc K: P_{K|S=s}(k|S=s)\leq \frac{1}{e}  \}
	\nonumber \end{align}
	and
	\begin{align}
	\mc K_{H}(s)=\{k\in \mc K: P_{K|S=s}(k|S=s)> \frac{1}{e}  \}.
	\nonumber \end{align}
	Clearly, it holds that $\lvert \mc K_{L}(s) \rvert +\lvert \mc K_{H}(s) \rvert =\lvert \mc K\rvert.$ 
	Let 
	\begin{align}
	P_{L}(s)=\sum_{k\in \mc K_{L}(s)} P_{K|S=s}(k|S=s)
	\nonumber \end{align}
	and
	\begin{align}
	P_{H}(s)=\sum_{k\in \mc K_{H}(s)} P_{K|S=s}(k|S=s).
	\nonumber \end{align}
	
	Notice first that
	\begin{align}
	1\geq P_{H}(s)> \lvert \mc K_{H}(s) \rvert \frac{1}{e}
	\nonumber \end{align}
	yielding
	\begin{align}
	\lvert \mc K_{H}(s) \rvert<e.
	\nonumber
	\end{align}
	Therefore,
	\begin{align}
	\lvert \mc K_{H}(s) \rvert\leq 2. \nonumber
	\end{align}
	Since $\lvert \mc K\rvert \geq 3,$ it follows that
	\begin{align}
	\lvert \mc K_L(s) \rvert=\lvert \mc K \rvert - \lvert \mc K_H(s)\rvert \geq 1. \nonumber
	\end{align}
	
	Now, it holds that
	\begin{align}
	&\mbb E \left[\ln^{2} P_{K|S=s}(K|S=s)|S=s\right]\nonumber \\
	&=\sum_{k\in \mc K_{L}(s)} P_{K|S=s}(k|S=s)\ln^2\frac{1}{P_{K|S=s}(k|S=s)}+ \sum_{k\in \mc K_{H}(s)} P_{K|S=s}(k|S=s)\ln^2\frac{1}{P_{K|S=s}(k|S=s)}.
	\label{termsinthesum} \end{align}
	We we will find appropriate upper-bound for each term in the right-hand side of \eqref{termsinthesum}. On the one hand, we have
	\begin{align}
	&\sum_{k\in \mc K_{L}(s)} P_{K|S=s}(k|S=s)\ln^2\left(\frac{1}{P_{K|S=s}(k|S=s)}\right) \nonumber \\
	&=P_{L}(s)\sum_{k\in \mc K_{L}(s)} \frac{P_{K|S=s}(k|S=s)}{P_{L}(s)}\ln^2\left(\frac{1}{P_{K|S=s}(k|S=s)}\right) \nonumber \\
	&\overset{(a)}{\leq} P_{L}(s)\ln^2\left( \sum_{k\in \mc K_{L}(s)} \frac{P_{K|S=s}(k|S=s)}{P_{L}(s)}\frac{1}{P_{K|S=s}(k|S=s)}\right) \nonumber \\
	&=P_{L}(s)\ln^{2}\frac{\lvert \mc K_{L}(s)\rvert}{P_{L}(s)},\nonumber \\
	\nonumber \end{align}
	where $(a)$ follows because $\ln^{2}(y)$ is concave in the range $y\geq e$ and because for any $k \in \mc K_{L}(s),$  $\frac{1}{P_{K|S=s}(k|S=s)}\geq e.$
	
	On the other hand, we have
	\begin{align}
	&\sum_{k\in \mc K_{H}(s)} P_{K|S=s}(k|S=s)\ln^2\frac{1}{P_{K|S=s}(k|S=s)} \nonumber\\
	&\overset{(a)}{\leq} \sum_{k\in \mc K_{H}(s)}  P_{K|S=s}(k|S=s) \ln^2(e) \nonumber \\
	&\leq 1,
	\nonumber \end{align}
	where $(a)$ follows because $\ln^2(1/y)$ is non-increasing in the range $0<y\leq 1$ and because $\frac{1}{e}<P_{K|S=s}(k|S=s)\leq 1$ for $k\in \mc K_{H}(s).$
	
	This implies using the fact that $\lvert \mc K \rvert \geq \lvert \mc K_{L}(s)\rvert\geq 1$
	that
	\begin{align}
	&\mbb E \left[\ln^{2} P_{K|S=s}(K|S=s)|S=s\right] \nonumber \\
	&\leq 1+ P_{L}(s)\ln^{2}\frac{\lvert \mc K_{L}(s)\rvert}{P_{L}(s)} \nonumber \\
	&=1+P_{L}(s)\left(\ln\left(\lvert\mc K_{L}(s)\rvert \right)+\ln\frac{1}{P_{L}(s)}  \right)^{2} \nonumber \\
	&\leq 1+P_{L}(s)\left(\ln\left(\lvert\mc K\rvert\right)+\ln\frac{1}{P_{L}(s)}  \right)^{2} \nonumber \\
	&= 1+P_{L}(s)\left(\ln\left(\lvert\mc K\rvert\right)^2+\ln^{2}\frac{1}{P_{L}(s)}+2\ln\left(\frac{1}{P_{L}(s)}\right) \ln\lvert\mc K\rvert \right) \nonumber \\
	&\overset{(a)}{\leq} 1+\ln\left(\lvert\mc K\rvert\right)^2+\frac{4}{e^2}+2\frac{1}{e}\ln\lvert\mc K\rvert \nonumber \\
	&= 1+\ln(2)^2\log\left(\lvert\mc K\rvert\right)^2+\frac{4}{e^2}+2\frac{\ln(2)}{e}\log\lvert\mc K\rvert \nonumber \\
	&\overset{(b)}{\leq} 1+\ln(2)^2n^{2}c^{2}+\frac{4}{e^2}+2\frac{\ln(2)}{e}nc \nonumber \\
	&<\infty,
 \end{align}
	where $(a)$ follows because $y\ln^2(1/y)$ and $y\ln(1/y)$ are maximized by $\frac{4}{e^2}$ and $\frac{1}{e}$ in the range $0<y\leq 1,$ respectively, and
	where $(b)$ follows because $\frac{\log \lvert \mc K \rvert}{n} \leq c$ 
	(from \eqref{cardinalitycorrelated}). 
	This proves Lemma \ref{upperboundvariance}.
\end{proof}
It follows using Lemma \ref{upperboundvariance} that $$\mathbb{E}\left[\log^2 P_{K|X_{1},\hdots,X_{J-1},Y_{J+1},\hdots,Y_{n},J,X_{J}=x}(K|X_{1},\hdots,X_{J-1},Y_{J+1},\hdots,Y_{n},J,X_{J}=x)|X_{J}=x\right]<\infty,$$ which implies that $\mathbb{E} \left[\left(\log P_{U|X_{J}=x}(U|X_{J}=x) \right)^2|X_{J}=x\right]<\infty.$ This completes the proof of Claim \ref{finitemean}.
Notice now that
{{\begin{align}
		H(K)&\overset{(a)}{=}H(K)-H(K|X^{n})\nonumber\\
		&=I(K;X^{n}) \nonumber\\
		&\overset{(b)}{=}\sum_{i=1}^{n} I(K;X_{i}|X_{1},\dots, X_{i-1}) \nonumber\\
		&=n I(K;X_{J}|X_{1},\dots, X_{J-1},J) \nonumber\\
		&\overset{(c)}{\leq }n I(U;X_{J}), \nonumber
		\end{align}}} where$(a)$ follows because $K=\Phi(X^n)$ and $(b)$ and $(c)$ follow from the chain rule for mutual information.
	
	We will show next that for some $\epsilon'(n)>0$
	
	$$ I(U;X_{J})-I(U;Y_{J})\leq C(W)+\epsilon'(n). $$
	
Applying Lemma \ref{lemma1} for $S=K$, $R=\varnothing$ with $V=(X_1,\hdots, X_{J-1},Y_{J+1},\hdots, Y_{n},J)$ yields
\begin{align}
&I(K;X^{n})-I(K;Y^{n}) 
\nonumber \\&=n[I(K;X_{J}|V)-I(K;Y_{J}|V)] \nonumber\\
&\overset{(a)}{=}n[I(KV;X_{J})-I(K;V)-I(KV;Y_{J})+I(K;V)] \nonumber\\
&\overset{(b)}{=}n[I(U;X_{J})-I(U;Y_{J})], 
\label{UhilfsvariableMIMO1}
\end{align}
where $(a)$ follows from the chain rule for mutual information and $(b)$ follows from $U=(K,V)$. \\
It results using (\ref{UhilfsvariableMIMO1}) that
\begin{align}
n[I(U;X_{J})-I(U;Y_{J})]
&=I(K;X^{n})-I(K;Y^{n}) \nonumber\\
&=H(K)-I(K;Y^{n})\nonumber \\ 
&=H(K|Y^n).
\label{star2MIMO2}
\end{align}
Next, we will show for some $\epsilon'(n)>0$ that
\begin{align}
\frac{H(K|Y^n)}{n}\leq C(W)+\epsilon'(n). \nonumber
\end{align}
We have
\begin{equation}
H(K|Y^{n})=I(K;Z^{n}|Y^{n})+H(K|Y^{n}Z^{n}).\label{boxed2}
\end{equation}
On the one hand, it holds that
\begin{align} 
I(K;Z^{n}|Y^{n})&\leq I(X^{n}K;Z^{n}|Y^{n}) \nonumber\\
& \overset{(a)}{\leq }I(T^n;Z^n|Y^{n})  \nonumber \\
& =  h(Z^n|Y^{n})- h(Z^n|T^n,Y^{n}) \nonumber \\
& \overset{(b)}{=}  h(Z^n|Y^{n})- h(Z^n|T^n) \nonumber \\
& \overset{(c)}{\leq }   h(Z^n)- h(Z^n|T^n) \nonumber \\
& = I(T^n;Z^n)  \nonumber \\
& \overset{(d)}{=} \sum_{i=1}^{n} I(Z_{i};T^n|Z^{i-1}) \nonumber \\
& = \sum_{i=1}^{n} h(Z_{i}|Z^{i-1})-h(Z_{i}|T^n,Z^{i-1}) \nonumber \\
& \overset{(e)}{=} \sum_{i=1}^{n} h(Z_{i}|Z^{i-1})-h(Z_{i}|T_{i}) \nonumber \\
& \overset{(f)}{\leq} \sum_{i=1}^{n} h(Z_{i})-h(Z_{i}|T_{i}) \nonumber \\
&=\sum_{i=1}^{n} I(T_{i};Z_{i}) \nonumber \\
&\leq n C(W), \label{part1}
\end{align}
where $(a)$ follows from the Data Processing Inequality for continuous RVs \cite{dataprocessing}, $(b)$ follows because $Y^{n}\circlearrow{X^{n}K}\circlearrow{T^n}\circlearrow{Z^{n}}$ forms a Markov chain, $(c)(f)$ follow because conditioning does not increase entropy, $(d)$ follows from the chain rule for mutual information and $(e)$ follows because $T_{1},\dots, T_{i-1},T_{i+1},\dots, T_{n},Z^{i-1} \circlearrow{T_{i}}\circlearrow{Z_{i}}$ forms a Markov chain. 
\color{black}
On the other hand, it holds that
\begin{align}
H(K|Y^{n},Z^{n})&\overset{(a)}{\leq } H(K|L) \nonumber \\
&\overset{(b)}{\leq } 1+\log\lvert \mathcal{K} \rvert \Pr[K\neq L] \nonumber \\
&\overset{(c)}{\leq }1+\epsilon c n, \label{part2}
\end{align}
where (a) follows from $L=\Psi(Y^{n},Z^{n})$ in \eqref{KLSISOcorrelated}, 
(b) follows from Fano's Inequality and (c) follows from \eqref{errorcorrelated} and \eqref{cardinalitycorrelated}.

It follows from \eqref{boxed2}, \eqref{part1} and \eqref{part2} that
\begin{align}
\frac{H(K|Y^n)}{n}\leq C(W)+\epsilon'(n),
\end{align}
where $\epsilon'(n)=\frac{1}{n}+\epsilon c.$
From \eqref{star2MIMO2}, we deduce that 
\begin{align}
I(U;X_{J})-I(U;Y_{J})\leq C(W)+\epsilon'(n).
\end{align}
Since the joint distribution of $X_{J}$ and $Y_{J}$ is equal to $P_{XY}$, $\frac{H(K)}{n}$ is upper-bounded by $I(U;X)$ subject to $I(U;X)-I(U;Y) \leq C(W) + \epsilon'(n)$ where $U \in \mc U$ and where $U \circlearrow{X} \circlearrow{Y}$ with \begin{equation}
\mathcal{U}=\Big\{U: \quad \mathbb{E}\left[ \log^2(P_{U|X}(U|X=x))|X=x \right]<\infty,\ \forall x \in \setx \Big\}. \nonumber
\end{equation}

 As a result, it holds using \eqref{ratecorrelated} that for sufficiently large $n,$ any achievable CR rate $H$ satisfies
\begin{align}
H < \underset{ \substack{U \in \mathcal{U} \\{\substack{U \circlearrow{X} \circlearrow{Y}\\ I(U;X)-I(U;Y) \leq \color{black}C(W)+\epsilon'(n)\color{black}}}}}{\sup} I(U;X)+\delta,
\label{righthandsideconverse}
\end{align}
with $\delta>0$ being the constant in \eqref{ratecorrelated}. In particular, we can choose $\epsilon$ and $\delta$ to be arbitrarily small positive constants such that the right-hand side of \eqref{righthandsideconverse} is equal to 
$$  \underset{ \substack{U \in \mathcal{U} \\{\substack{U \circlearrow{X} \circlearrow{Y}\\ I(U;X)-I(U;Y) \leq \color{black}C(W)\color{black}}}}}{\sup} I(U;X)+\epsilon'',   $$
for $n\rightarrow \infty,$ with $\epsilon''$ being an arbitrarily small positive constant.
This completes the converse proof.
\section{Conclusion}
CR generation has striking applications in the identification scheme, a new approach in communications that is highly relevant in 6G Communication. Indeed, in contrast to Shannon message transmission, the resource CR allows a significant increase in the identification capacity of channels. For this reason, 
CR generation for future communication networks is a central research question in large 6G research projects. It is also worth mentioning that CR is highly relevant in the modular coding scheme for secure communication and a useful resource in coding over AVCs.
In this paper, we investigated the problem of CR generation
from correlated sources with countable alphabets aided by one-way communication over noisy memoryless channels. We established a single-letter expression for the CR capacity. The coding
scheme for CR generation that we proposed is  based on the same type of binning as in the Wyner-Ziv problem. The novelty lies in extending the Wyner-Ziv coding scheme to infinitely countable alphabets. As a future work, it would be interesting to investigate the problem of CR  generation from correlated sources with arbitrary joint distribution. \label{conclusion}
\section{Acknowledgments}
H. Boche was supported in part by the German Federal Ministry of Education and Research (BMBF) within the national initiative on 6G Communication Systems through the research hub 6G-life under Grant 16KISK002, within the national initiative on Post Shannon Communication
(NewCom) under Grant 16KIS1003K. It has further received funding by the German Research Foundation (DFG) within Germany’s Excellence Strategy EXC-2092 – 390781972. M. Wiese was supported by the Deutsche Forschungsgemeinschaft (DFG, German Research Foundation) within the Gottfried Wilhelm Leibniz Prize under Grant BO1734/20-1, and within Germany's Excellence Strategy EXC-2111-390814868 and EXC-2092 CASA-390781972. C.\ Deppe was supported in part by the German Federal Ministry of Education and Research (BMBF) under Grant 16KIS1005 and in part by the German Federal Ministry of Education and Research (BMBF) within the national initiative on 6G Communication Systems through the research hub 6G-life under Grant 16KISK002. W. Labidi and R. Ezzine were supported by the German Federal Ministry of Education and Research (BMBF) under Grant 16KIS1003K.

\bibliographystyle{IEEEtran}
\bibliography{definitions,references}

\IEEEtriggeratref{4}



\end{document}